\documentclass[onecolumn]{IEEEtran}

\ifCLASSINFOpdf
\else
  \usepackage[dvips]{graphicx}
  \graphicspath{{../eps/}}
   \DeclareGraphicsExtensions{.eps}
\fi
 \RequirePackage{amssymb} 
\usepackage{amsmath, amssymb}

\usepackage{graphicx}
\usepackage{tikz}
\usetikzlibrary{shapes,arrows,fit}

\usepackage{caption}
\usepackage{subcaption}
\usepackage{cite}
\usepackage{psfrag}

\newcommand{\mc}[1]{\mathcal{#1}}

\newcommand{\ve}[1]{\mathbf{#1}}

\newcommand{\hY}{\hat Y_r}
\newcommand{\hy}{\hat y_r}

\newtheorem{proof}{Proof}
\newtheorem{prop}{Proposition}
\newtheorem{lemma}{Lemma}
\newtheorem{cor}{Corollary}

\begin{document}

\sloppy

\tikzstyle{block} = [draw, fill=blue!5, rectangle, rounded corners,
    minimum height=2em, minimum width=6em]
\tikzstyle{input} = [draw, fill=green!5, rectangle, rounded corners,
    minimum height=2em, minimum width=4em]
\tikzstyle{chan} = [draw, fill=red!5, rectangle, rounded corners,
    minimum height=2em, minimum width=4em]
\tikzstyle{output} = [coordinate]
\tikzstyle{pinstyle} = [pin edge={to-,thin,black}]
\tikzstyle{block2} = [draw, rectangle, 
    minimum height=1.5em, minimum width=2em]
\tikzstyle{sum} = [draw, circle, inner sep=0pt, node distance=1cm]

\newcommand{\tone}{\ensuremath{T_1}}
\newcommand{\ttwo}{\ensuremath{T_2}}
\newcommand{\tr}{\ensuremath{r}}

\newcommand{\w}{\ensuremath{W}}
\newcommand{\wone}{\ensuremath{W_1}}
\newcommand{\wtwo}{\ensuremath{W_2}}
\newcommand{\wonehat}{\ensuremath{\hat{W}_1}}
\newcommand{\wtwohat}{\ensuremath{\hat{W}_2}}

\newcommand{\x}{\ensuremath{X}}
\newcommand{\xone}{\ensuremath{X_1}}
\newcommand{\xtwo}{\ensuremath{X_2}}
\newcommand{\xr}{\ensuremath{X_r}}

\newcommand{\Y}{\ensuremath{Y}}
\newcommand{\yr}{\ensuremath{Y_r}}
\newcommand{\yone}{\ensuremath{Y_1}}
\newcommand{\ytwo}{\ensuremath{Y_2}}

\newcommand{\zone}{\ensuremath{Z_1}}
\newcommand{\ztwo}{\ensuremath{Z_2}}
\newcommand{\zr}{\ensuremath{Z_r}}

\newcommand{\Pone}{\ensuremath{P_1}}
\newcommand{\Ptwo}{\ensuremath{P_2}}
\newcommand{\Prel}{\ensuremath{P_r}}

\newcommand{\Nr}{\ensuremath{N_r}}
\newcommand{\None}{\ensuremath{N_1}}
\newcommand{\Ntwo}{\ensuremath{N_2}}

%
\title{Scalar Quantize-and-Forward for Symmetric Half-duplex Two-Way Relay Channels }

\author{\IEEEauthorblockN{ Michael Heindlmaier, Onurcan \.{I}\c{s}can, Christopher Rosanka}

\IEEEauthorblockA{Institute for Communications Engineering, Technische Universit\"at M\"unchen, 
Munich, Germany\\
Email: michael.heindlmaier@tum.de, onurcan.iscan@tum.de, ch.rosanka@mytum.de}
}

\maketitle

\IEEEpeerreviewmaketitle

\begin{abstract}
THIS PAPER IS ELIGIBLE FOR THE STUDENT PAPER AWARD\\
Scalar Quantize \& Forward (QF) schemes are studied for the Two-Way Relay Channel.
Different QF approaches are compared in terms of rates as well as relay and decoder complexity.
A coding scheme not requiring Slepian-Wolf coding at the relay is proposed and properties of the corresponding sum-rate optimization problem are presented. A numerical scheme similar to the Blahut-Arimoto algorithm is derived that guides optimized quantizer design.
The results are supported by simulations.

\end{abstract}  

\section{Introduction}

Consider a communication system where two nodes \tone\, and \ttwo\, wish to communicate to each other with the help of a relay \tr\, and there is no direct link between \tone\, and \ttwo. This scenario is known as a separated two-way relay channel (TWRC) \cite{gunduz2008rate,rankov2006achievable} and it incorporates challenging problems such as multiple access, broadcast, and coding with side information. 

In this work, we focus on Quantize \& Forward (QF) relaying: The relay maps its received (noisy) signal 
to a quantization index by using a quantizer function $\mathcal{Q}(.)$. The index is then digitally transmitted to the destination nodes through the downlink channels. 
We use information theoretic arguments to find quantizers that almost maximize the sum-rate, 
an approach that has been proposed in \cite{winkelbauer2011soft} and \cite{zeitler2012} in a similar context.

In general, \emph{vector} quantizers give the best performance, but under certain conditions \emph{scalar} quantizers almost maximize the sum-rate. Scalar quantizers are attractive because of their design and implementation simplicity.

Our main focus is the symmetric TWRC, where both users' channel qualities are the same, both in the uplink and downlink.
We describe the system model in Sec.~\ref{sec:sysmod}, and in Sec.~\ref{sec:theory} we compare different achievable rate regions for the TWRC. We introduce a new rate region that is smaller than previous regions, but almost as large for the symmetric TWRC. The achievability scheme does not require the relay to employ Slepian-Wolf Coding. 
In Sec.~\ref{sec:QuantizerDesign}, we look at the sum-rate optimization and quantizer design problems and propose a numerical solution. Examples of optimal quantizers are given and optimal time sharing parameters are calculated. In Sec.~\ref{sec:sims} we evaluate the performance of the optimized system by simulations. Sec.~\ref{sec:conclusion} concludes our work and gives future directions.
  
\section{System Model}
\label{sec:sysmod}
The system has two source nodes \tone\, and \ttwo\, that exchange their independent messages $\wone\ \in \{1,2,\ldots, 2^{nR_1}\}$, $\wtwo\ \in \{1,2,\ldots, 2^{nR_2}\}$ in $n$ channel uses through a relay node \tr. The source nodes cannot hear each other, so communication is possible only through the relay. The communication is split into two phases: In the multiple access (MAC) phase with $n_\text{MAC}$ channel uses, both source nodes encode their messages \wone\, and \wtwo\ to the MAC channel inputs $X_1^{n_\text{MAC}} $  and $X_2^{n_\text{MAC}}$, respectively, with $X_{1,t} \in \mc X_1$, $X_{2,t} \in \mc X_2$. 
Define $\alpha = n_\text{MAC} / n$ as the time fraction of this first phase.
The relay receives 
\begin{equation}
 Y_{r,t} = X_{1,t} + X_{2,t} + Z_{r,t}, \quad t=\{1,2,\ldots,n_\text{MAC}\},
\end{equation}
where $Z_{r,t} \sim \mc N(0,N_r)$ and $\mathbb E\{X_{1,t}^2 \} \leq P_1$, $\mathbb E\{X_{2,t}^2 \} \leq P_2$.

The relay maps $Y_r^{n_\text{MAC}}$ to a quantized representation $\hY^{n_\text{MAC}}$ with symbol alphabet $\mc{\hat Y}_r$. The quantizer index is $q = \mc Q(Y_r^{n_\text{MAC}})$. 
During the Broadcast (BC) phase with $n_\text{BC} = n-n_\text{MAC}$ channel uses, the relay transmits the codeword $X_r^{n_\text{BC}}(q)$. 
The received signals at \tone\, and \ttwo\, are:
\begin{align}
Y_{j,t} = X_{r,t} + Z_{j,t} , \quad t=\{n_\text{MAC}+1,\ldots,n\},
\end{align}
for $j\in\{1,2\}$, $\mathbb E\{X_{r,t}^2 \} \leq P_r$ and $Z_{j,t} \sim \mc N(0,N_j)$.
Nodes \tone\, and \ttwo\ decode $W_2$ and $W_1$, respectively, by using their own message as side information. Fig. \ref{fig:block} depicts the system setup.
In the following, we often omit the time index $t$ if we refer to a single channel use.
In the next section we review different coding schemes and compare their performance.

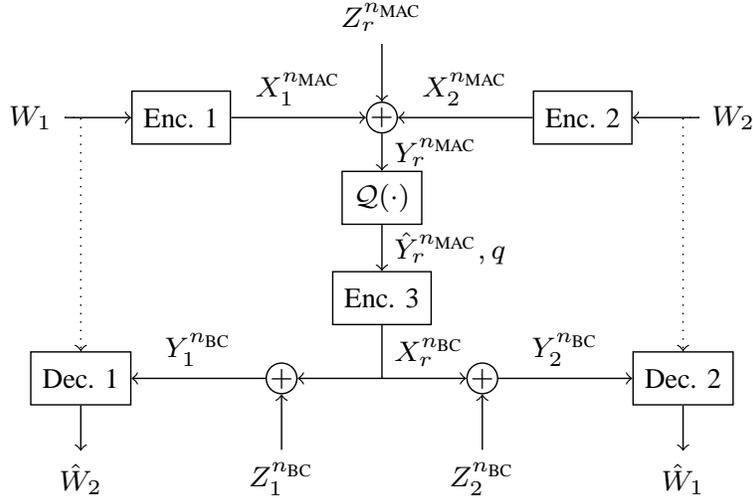
\begin{figure}[t]
\begin{center}
\resizebox{0.6\columnwidth}{!}{
  \tikzstyle{every node}=[font=\footnotesize ]
\begin{tikzpicture}[auto, node distance=2cm]
  \node (W1) {$\wone$};
  \node [block2, right of=W1, node distance = 15mm] (E1) {Enc. 1};
  \node [sum,  right of=E1, node distance =20mm] (summing) {$+$};
  \node [block2, right of=summing, node distance = 20mm] (E2) {Enc. 2};
  \node [right of=E2, node distance = 15mm] (W2) {$\wtwo$};
  \node [block2, below of=summing, node distance = 8mm] (Q) {$\mc Q(\cdot)$};
  \node [block2, below of=Q, node distance = 10mm] (E3) {Enc. 3};  
  \node [coordinate, node distance = 8mm,below of=E3] (corner){};
  \node [sum, left of=corner, node distance =10mm] (sum1) {$+$};  
  \node [block2,node distance = 20mm,left of = sum1] (D1) {Dec. 1};
  \node [sum,  right of=corner, node distance =10mm] (sum2) {$+$};
  \node [block2,node distance = 20mm,right of = sum2] (D2) {Dec. 2};
  \node [below of=D1, node distance = 10mm] (W2h) {$\wtwohat$};
  \node [below of=D2, node distance = 10mm] (W1h) {$\wonehat$};
  \node [above of=summing, node distance = 10mm] (Nr) {$Z_r^{n_{\text{MAC}}}$};  
  \node [below of=sum1, node distance = 10mm] (N1) {$Z_1^{n_{\text{BC}}}$};
  \node [below of=sum2, node distance = 10mm] (N2) {$Z_2^{n_{\text{BC}}}$};    
  \node [coordinate, node distance = 15mm,below of=W1] (corner2){};
  \node [coordinate, node distance = 15mm,right of=corner2] (corner3){};
  \node [coordinate, node distance = 15mm,below of=W2] (corner4){};
  \node [coordinate, node distance = 15mm,left of=corner4] (corner5){};
  \draw [->] (W1) --   (E1);
  \draw [->] (W2) -- (E2);
  \draw [->] (E1) -- node[name=X1]{$X_1^{n_{\text{MAC}}}$} (summing);
  \draw [->] (E2) -- node[name=X2, above]{$X_2^{n_{\text{MAC}}}$} (summing);
  \draw [->] (summing) -- node[name=Yr]{$Y_r^{n_{\text{MAC}}}$} (Q);
  \draw [->] (Q) -- node[name=Yrhat]{$\hat{Y}_r^{n_{\text{MAC}}}, q$} (E3);
  \draw [-] (E3) -- node[name=Xr]{$X_r^{n_{\text{BC}}}$} (corner);
  \draw [->] (corner) -- (sum1);
  \draw [->] (corner) -- (sum2);
  \draw [->] (sum1) -- node[above, name=Y1]{$Y_1^{n_{\text{BC}}}$} (D1);
  \draw [->] (sum2) -- node[name=Y2]{$Y_2^{n_{\text{BC}}}$} (D2);
  \draw [->] (D1) -- (W2h);
  \draw [->] (D2) -- (W1h);
  \draw [->] (Nr) -- (summing);
  \draw [->] (N1) -- (sum1);
  \draw [->] (N2) -- (sum2);
  \draw [dotted,->] (W1)  -| (D1);
  \draw [dotted,->] (W2)  -| (D2);
\end{tikzpicture}
}
\end{center}
\caption{System block diagram.}
\label{fig:block}
\end{figure}

\section{Achievable Rates}
\label{sec:theory}

\subsection{Noisy Network Coding / Joint Decoding}
In \cite{schnurr2008coding}, an achievable rate region was derived that matches the rates achieved with Noisy Network Coding (NNC) \cite{lim2011noisy}.
The closure of the achievable rate region is given by the set $\mc R_{\text{NNC}} \subset \mathbb R_+^2$ of rate tuples $(R_1,R_2)$ satisfying
\begin{align}
 R_1 & \leq \alpha I(X_1; \hY|X_2,U) \nonumber \\
 R_1 & \leq (1-\alpha) I(X_r;Y_2) - \alpha I(Y_r;\hY|X_1, X_2, U) \nonumber \\
 R_2 & \leq \alpha I(X_2; \hY|X_1, U) \nonumber \\
 R_2 & \leq (1-\alpha) I(X_r;Y_1) - \alpha I(Y_r;\hY|X_1, X_2, U) 
\end{align}
for some $p(u) p(x_1|u) p(x_2|u) p(y_r|x_1,x_2) p(\hy|y_r)$ and $p(x_r) p(y_1,y_2|x_r)$ and $\alpha >0$.
It suffices to consider $|\mc{\hat Y}_r| \leq |\mc Y_r| + 2$,  $|\mc U|\leq 3$.

\subsubsection{Receivers}
To achieve a rate tuple in $\mc R_{\text{NNC}}$, the decoder must jointly decode the BC code and the MAC code in a single stage decoder using its own message as side information. The quantization index $q$ is not required to be decoded. 
The decoder structure is shown in Fig.~\ref{fig:decNNC}.

\subsubsection{Gaussian Case}
For the Gaussian case, $X_1 \sim \mc N(0,P_1)$, $X_2 \sim \mc N(0,P_2)$, $X_r \sim \mc N(0,P_r)$ and $Y_r \sim \mc N(0,P_1+P_2+N_r)$.
We choose a (not necessarily optimal) quantizer yielding $\hY = Y_r + \hat Z$, where the quantization noise $\hat Z \sim \mc N(0,\hat N)$ is independent of $Y_r$.
Define $C(x):=\frac{1}{2} \log (1+x)$. 
With $\mc U = \emptyset$, the achievable rate region becomes
\begin{align}
 R_1 & \leq \min\left\{ \alpha C\left(  \frac{P_1}{N_r + \hat N}\right), (1-\alpha) C\left(  \frac{P_r}{N_2} \right) - \alpha C \left(  \frac{N_r}{\hat N} \right)\right\} \nonumber \\
 R_2 & \leq \min\left\{ \alpha C\left(  \frac{P_2}{N_r + \hat N}\right), (1-\alpha) C \left(  \frac{P_r}{N_1} \right) - \alpha C \left(  \frac{N_r}{\hat N} \right)\right\} .
\label{eq:NNCGauss}
\end{align}

\begin{figure}
{
\begin{subfigure}[t]{0.33\linewidth}
\centering
\begin{tikzpicture}[auto,node distance=2cm]
  \tikzstyle{every node}=[font=\footnotesize ]
  \node (W1) {$\wone$};
  \node [below of=W1, node distance = 8mm, label=below:{$\hat{X}_2^{n_{\text{MAC}}}$}] (W2h) {$\wtwohat$};
  \node [block2, right of= W2h, node distance = 15mm] (Dec1) {Joint Dec.};
  \node [right of=Dec1, node distance = 14mm] (Y1) {$Y_1^{n_{\text{BC}}}$};

    \node [draw,rectangle, rounded corners, black, dashed, fit = (Dec1) , label=below:{Dec. 1}](dd) {};
   \draw [->] (Y1) --  (Dec1);
   \draw [->] (Dec1) -- (W2h);

   \draw [dotted,->] (W1)  -| (Dec1);

\end{tikzpicture}
 \caption{Decoder for $\mc R_{\text{NNC}}$. }
 \label{fig:decNNC}
\end{subfigure}
\hspace{0.5cm}
\begin{subfigure}[t]{0.67\linewidth}
\centering
\begin{tikzpicture}[auto,node distance=2cm]
  \tikzstyle{every node}=[font=\footnotesize ]
  \node (W1) {$\wone$};
  \node [below of=W1, node distance = 8mm, label=below:{$\hat{X}_2^{n_{\text{MAC}}}$}] (W2h) {$\wtwohat$};
  \node [block2, right of= W2h, node distance = 15mm,text width=1cm,align=center] (MACDec) {MAC Dec.};
  \node [block2, right of= MACDec, node distance = 15mm,text width=0.8cm,align=center] (BCDec) {BC Dec.};
  \node [right of=BCDec, node distance = 12mm] (Y1) {$Y_1^{n_{\text{BC}}}$};

  \node [draw,rectangle, rounded corners, black, dashed, fit = (MACDec) (BCDec) , label=below:{Dec. 1}] (dd){};
   \draw [->] (Y1) --  (BCDec);
   \draw [->] (MACDec) -- (W2h);
   \draw [->] (BCDec) -- node[name=q] {$\hat q$} (MACDec);
   \draw [dotted,->] (W1)  -| (BCDec);
   \draw [dotted,->] (W1)  -| (MACDec);
\end{tikzpicture}
\caption{Decoder  for  $\mc R_{\text{CF}}$.}
\label{fig:decCF}
\end{subfigure}

\begin{subfigure}[b]{0.90\linewidth}
\centering 
\begin{tikzpicture}[auto,node distance=2cm]
  \tikzstyle{every node}=[font=\footnotesize ]
  \node (W1) {$\wone$};
  \node [below of=W1, node distance = 8mm, label=below:{$\hat{X}_2^{n_{\text{MAC}}}$}] (W2h) {$\wtwohat$};
  \node [block2, right of= W2h, node distance = 15mm,text width=1cm,align=center] (MACDec) {MAC Dec.};
  \node [block2, right of= MACDec, node distance = 15mm,text width=0.8cm,align=center] (BCDec) {BC Dec.};
  \node [right of=BCDec, node distance = 15mm] (Y1) {$Y_1^{n_{\text{BC}}}$};

  \node [draw,rectangle, rounded corners, black, dashed, fit = (MACDec) (BCDec) , label=below:{Dec. 1}] (dd){};
   \draw [->] (Y1) --  (BCDec);
   \draw [->] (MACDec) -- (W2h);
   \draw [->] (BCDec) -- node[name=q] {$\hat q$} (MACDec);
   \draw [dotted,->] (W1)  -| (MACDec);
%
%
%
\end{tikzpicture}
\caption{Decoder  for $\mc R_{\text{NoSW}}$.}
\label{fig:decNoSW}
\end{subfigure} 
}
\caption{Decoder structure for different schemes.}
\label{fig:decoder}
\end{figure}
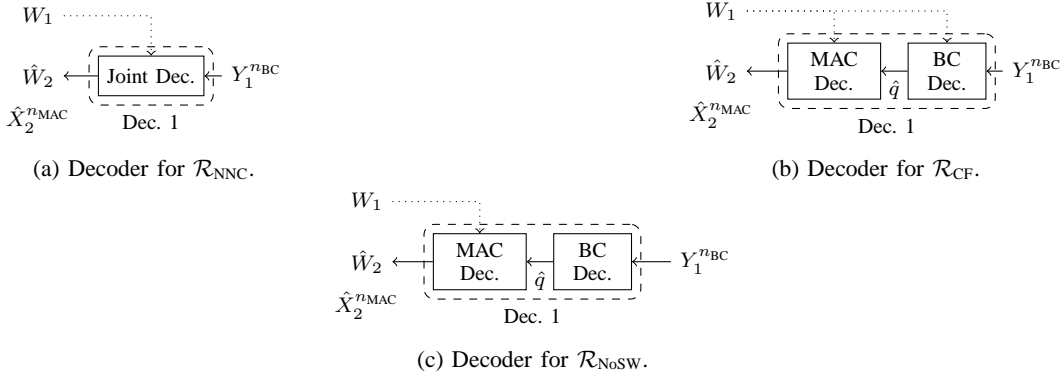

\subsection{Compress \& Forward}
\label{sec:CF}
In the spirit of classic Compress \& Forward (CF), the authors of \cite{schnurr2007achievable,kim2011achievable} derive an achievable rate region generalizing \cite{rankov2006achievable} using ideas from \cite{tuncel2006slepian}.
The achievable rate region is the set $\mc R_{\text{CF}} \subset \mathbb R_+^2$ of rate tuples $(R_1,R_2)$ satisfying
\begin{align}
R_1 \leq \alpha I(X_1;\hY|X_2,U) ,& \quad
R_2 \leq \alpha I(X_2;\hY|X_1,U) \nonumber 
\\
\text{for}\quad \alpha I(Y_r ; \hY | X_2,U) &\leq (1-\alpha) I(X_r; Y_2) \nonumber \\
\alpha I(Y_r ; \hY | X_1,U) &\leq (1-\alpha) I(X_r; Y_1) 
\end{align}
for some $p(u)p(x_1|u)p(x_2|u)p(y_r|x_1, x_2)p(\hy|y_r)$ and $p(x_r)p(y_1,y_2|x_r)$, $\alpha > 0$. 
It suffices to consider $|\mc{\hat Y}_r| \leq |\mc Y_r| + 3$,  $|\mc U|\leq 4$.

\subsubsection{Receivers}
The coding scheme of \cite{schnurr2007achievable} requires reliable decoding of the quantization index $q$ at the receiver. For that, the BC code is decoded using the own message as a priori knowledge. Knowing $q$, the desired message is decoded, again using the own message as side information.
The structure of this decoder is shown in Fig.~\ref{fig:decCF}.

\subsubsection{Gaussian Case}
With the same assumptions as before, i.e. $\hY = Y_r + \hat Z$, one obtains:
\begin{align}
 R_1 \leq  \alpha C\left(\frac{P_1}{N_r + \hat N} \right)  ,& \quad
 R_2 \leq  \alpha C \left(  \frac{P_2}{N_r + \hat N} \right) \nonumber
\\
\text{for} \quad \alpha C\left(  \frac{P_1+N_r}{\hat N} \right) \leq & (1-\alpha) C \left(  \frac{P_r}{N_2} \right) \nonumber\\
\alpha C\left(  \frac{P_2+N_r}{\hat N} \right) \leq & (1-\alpha) C \left(  \frac{P_r}{N_1} \right).
\end{align}

\subsection{Neglecting Side Information in the Downlink}
\label{sec:scheme_NoSW}
The coding scheme for CF requires Slepian-Wolf (SW) coding \cite{slepian1973noiseless} at the relay to reduce the downlink rate. In the symmetric case we do not expect this reduction to be substantial. 
We are thus interested in schemes without SW coding in the BC phase.
Using random coding arguments, the achievable rates are given by the set $\mc R_{\text{NoSW}} \subset \mathbb R_+^2$ of rate tuples $(R_1,R_2)$ satisfying
\begin{align}
 R_1 \leq \alpha I(X_1;\hY|X_2, U), \quad R_2 \leq \alpha I(X_2;\hY|X_1, U) \nonumber 
\\
\text{for } \alpha I(Y_r ; \hY |U ) \leq (1-\alpha) \min\{ I(X_r; Y_2), I(X_r; Y_1)\}  \label{eq:RRnoSI}
\end{align}
for some $p(u) p(x_1|u) p(x_2|u) p(y_r|x_1,x_2) p(\hy|y_r)$ and $p(x_r)p(y_1,y_2|x_r)$, $\alpha > 0$.
Similarly, we have $|\mc U| \leq 4$ and $|\mc{\hat Y}_r| \leq |\mc Y_r| + 3$.
A proof of this claim can be found in 
Appendix \ref{sec:achievability}.

\subsubsection{Receivers}
The structure of the decoder is shown in Fig.~\ref{fig:decNoSW}.
Similar to the scheme in Sect.~\ref{sec:CF}, two decoding stages are required: First the the BC code is decoded, revealing $q$ reliably. Then $q$ is used to obtain the desired message from the MAC code. In contrast to before, the own message is used only in the MAC decoder.

\subsubsection{Gaussian Case}
The rate region is described by
\begin{align}
 R_1 \leq  \alpha C\left( \frac{P_1}{N_r + \hat N} \right) , &\quad
 R_2 \leq  \alpha C \left(  \frac{P_2}{N_r + \hat N} \right) \nonumber
\\
\alpha C \left(  \frac{P_1+P_2+N_r}{\hat N} \right) \leq  & (1-\alpha) \min \left\{ C\left(  \frac{P_r}{N_2} \right), C \left( \frac{P_r}{N_1} \right)   \right\} \nonumber 
\end{align}

\subsection{Sum-Rate Comparison for Gaussian Case}
\label{sec:sum_rate_comparison}
Note that in general $\mc R_{\text{NoSW}} \subset \mc R_{\text{CF}} \subset \mc R_{\text{NNC}}$.
We focus on the maximal sum rate $R_1 + R_2$ for each scheme in the symmetric Gaussian case $P_1 = P_2=P$, $N_1 = N_2=N$.
This requires to jointly optimize over the quantization noise variance $\hat N$ and time allocation $\alpha$. 
Formally, define 
\begin{align}
 R^{\text{sum}}_{\text{NNC}}(\alpha):=\max_{\hat N} (R_1+R_2), \quad (R_1,R_2)\in \mc R_{\text{NNC}} \nonumber
\end{align}
 and similarly for $R^{\text{sum}}_{\text{CF}}(\alpha)$ and $R^{\text{sum}}_{\text{NoSW}}(\alpha)$.
The optimal sum rate is
\begin{align}
 R^{\text{sum}}(\alpha) = 2 \alpha \mc C\left(\frac{P}{N_r + \hat N^*(\alpha)} \right) \label{eq:Rsum_alpha}
\end{align}
for all three schemes. The only difference is the optimal value of $\hat N$ for a given $\alpha$, denoted by $\hat N^*(\alpha)$.
For $\mc R_{\text{NoSW}}$, we have
$$\hat N^*_{\text{NoSW}}(\alpha) = \frac{2 P + N_r}{\left(1 + P_r/N \right)^{(1-\alpha)/\alpha}-1}$$
because the sum rate is decreasing in $\hat N$ and $\hat N^*_{\text{NoSW}}(\alpha)$ is the smallest variance satisfying the constraints.
Similarly, for $\mc R_{\text{CF}}$, we have
$$\hat N^*_{\text{CF}}(\alpha) = \frac{P + N_r}{\left(1 + P_r/N \right)^{(1-\alpha)/\alpha}-1}.$$

As $\hat N^*_{\text{CF}}(\alpha) < \hat N^*_{\text{NoSW}}(\alpha)$, we have $R^{\text{sum}}_{\text{CF}}(\alpha) > R^{\text{sum}}_{\text{NoSW}}(\alpha)$.
For $\mc R_{\text{NNC}}$, the rate expressions for NNC in Eq.~(\ref{eq:NNCGauss}) are either increasing or decreasing in $\hat N$. The maximum with respect to $\hat N$ is thus found at the crossing point of both expressions. It is not hard to show that $\hat N^*_{\text{NNC}}(\alpha)=\hat N^*_{\text{CF}}(\alpha)$ (see 
Appendix \ref{sec:optimal_hatN}
for a derivation). This means, given the assumption that $(Y_r,\hY)$ are jointly Gaussian, NNC does not provide higher sum rates than CF in the symmetric case.

Fig.~\ref{fig:SumRateOverSNR_offset} shows achievable sum rates over SNR for the symmetric Gaussian case. 
For each curve, the value of $\alpha$ was chosen to maximize the sum rate.
As expected for this setup, the difference between $R^{\text{sum}}_{\text{CF}}(\alpha^*)=R^{\text{sum}}_{\text{NNC}}(\alpha^*)$ and $R^{\text{sum}}_{\text{NoSW}}(\alpha^*)$ is small. 
The scheme corresponding to $R^{\text{sum}}_{\text{NoSW}}(\alpha^*)$ requires less complex relay operations.
Therefore we focus on this scheme, accepting a slightly smaller achievable sum rate.
Also note that we assume that $\mc U = \emptyset$.
\begin{figure}[t]
\centering
\vspace*{-3mm}
\includegraphics[width=0.6\columnwidth]{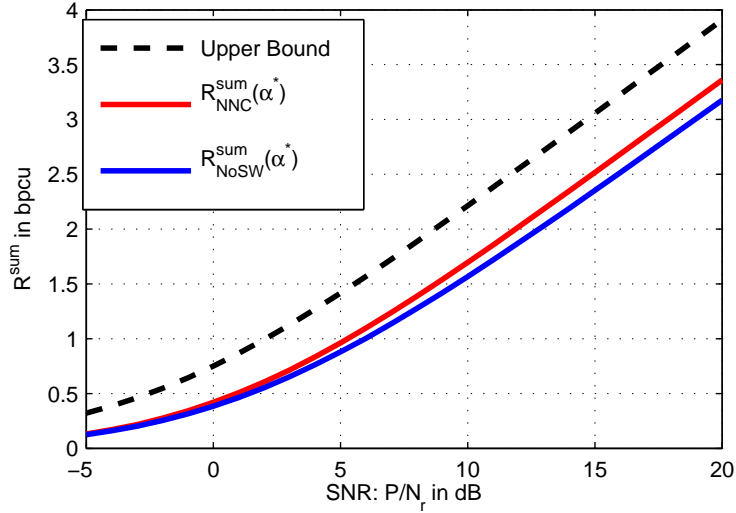}
\caption{Sum rates $R^{\text{sum}}_{\text{CF}}(\alpha^*)=R^{\text{sum}}_{\text{NNC}}(\alpha^*)$ and $R^{\text{sum}}_{\text{NoSW}}(\alpha^*)$ as a function of SNR$:=P/N_r$ and $P_r = P +8.5$dB. The upper bound is $\min\{2\alpha \mc C(P/N_r), 2(1-\alpha) \mc C(P_r/N)\}$.}
\label{fig:SumRateOverSNR_offset}
\end{figure}

\section{Quantizer Design}
\label{sec:QuantizerDesign}

\subsection{Sum-Rate optimization for $\mc R_{\text{NoSW}} $}
\label{sec:quant_design_sum_rate}
To guide the quantizer design for DMCs and fixed discrete input distributions, we want to find the optimal conditional probability mass function (pmf) $p(\hy|y_r)$ and time sharing coefficient $\alpha$ to optimize the sum-rate. 
With $C~:=~\min\{ I(X_r; Y_2), I(X_r; Y_1)\}$ as the downlink capacity, this problem can be stated as
\begin{align}
 R^{\text{sum}}_{\text{NoSW}} := \sup_{\alpha, p(\hy|y_r)} &\alpha \left( I(X_1;\hY|X_2) + I(X_2;\hY|X_1) \right) \nonumber \\
 \text{subject to: } \quad &\alpha I(Y_r ; \hY) \leq (1-\alpha) C.
\label{eq:sum_rate_opt_noSI}
\end{align}
Abbreviate $p(\hy|y_r)$ by $p$ and let $p_{ij}:=p(\hat{y}_{ri}|y_{rj})$:
Denote the objective as $J(p):=I(X_1;\hY|X_2) + I(X_2; \hY|X_1)$ and define the function 
\begin{align}
 I(C) := \sup_{p(\hy|y_r): I(Y_r; \hY) \leq C}  J(p).
\label{eq:def_I(C)}
\end{align}
Problem (\ref{eq:sum_rate_opt_noSI}) can be stated as 
\begin{equation}
  R^{\text{sum}}_{\text{NoSW}} = \sup_{\alpha} \alpha I\left(\frac{1-\alpha}{\alpha} C  \right). \label{eq:I_dependent_on_alpha}
\end{equation}
Some properties of $I(C)$ are as follows:


\begin{enumerate}

\item The functions $I(X_1;\hY|X_2) + I(X_2; \hY|X_1)$ and $ I(Y_r;\hY)$ are
convex in $p(\hy|y_r)$ \cite[Theorem 2.7.4]{cover2006elements}. 
Program (\ref{eq:def_I(C)}) is thus a convex maximization, a difficult problem in general.
From the \emph{maximum principle} \cite[Cor. 32.3.2]{rockafellar1997convex}, it follows that for the optimal $p(\hy|y_r)$, $ I(Y_r;\hY) = C$, for $0 \leq C \leq H(Y_r)$.

\item \label{property:concave} $I(C)$ is an increasing and concave function in
$C$, for $0 \leq C \leq H(Y_r)$. The proof is along the lines of \cite{witsenhausen1975conditional}.
We refer to 
Appendix \ref{sec:derivation_propertiesIC}
 for a more detailed derivation. 
\end{enumerate}

Fig.~\ref{fig:I(C)_matlabexample}  shows one example of $I(C)$. $L = |\mc{\hat Y}_r|$ represents the number of different quantization levels.
One can see that it suffices to consider only a relatively small $L$. For example, using a mapping with $L=8$ quantization levels instead of $L=16$ causes a rate reduction of less then 0.03 in $I(C)$.

\begin{figure}[t]
 \centering
\vspace*{-3mm}
\includegraphics[width=0.6\columnwidth]{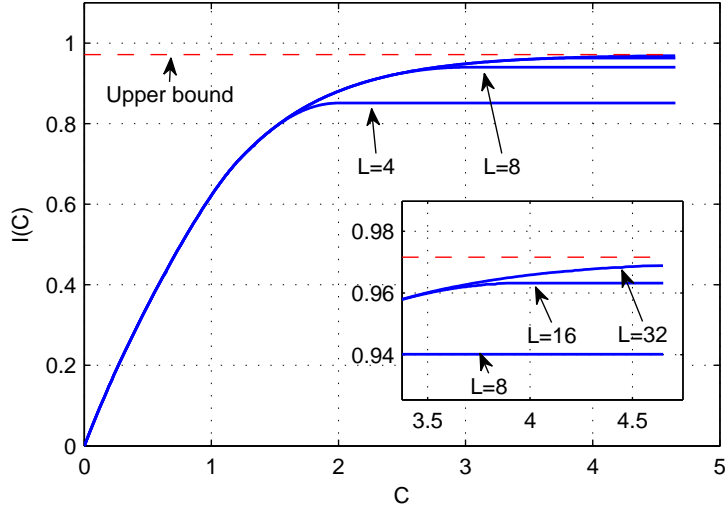}
\caption{$I(C)$ and upper bound $I(X_1;Y_r|X_2)+I(X_2;Y_r|X_1)$ for BPSK modulation at both users and $P=0$dB, $N=0$dB, for different numbers of quantization levels $L=|\mc{\hat Y}_r|$.}
\label{fig:I(C)_matlabexample}
\end{figure}

\subsection{Computing the function $I(C)$}
\label{sec:computeI(C)}
To solve the problem for $I(C)$ in Eq.~(\ref{eq:def_I(C)}), similar to \cite[Section 10]{cover2006elements} we write the Lagrangian:
\begin{align}
 \mc L(p, \lambda, \nu_1,  \ldots,\nu_{L}) 
=J(p) - \lambda I(Y_r; \hY) + \sum_j \nu_j \sum_i p_{ij}, \nonumber
\end{align}
$\lambda \geq 0$, where the last constraints force $p(\hy|y_r)$ to be a valid pmf.
The KKT conditions require
\begin{align}
\frac{\partial \mc L }{\partial p_{ij}} = \frac{\partial J}{\partial p_{ij}} - \lambda p(y_{rj}) \log \frac{p_{ij}}{p(\hat y_{ri}) }  + \nu_j \stackrel{!}{=} 0, \quad \forall~i,j.
\end{align}
With the substitution
\begin{align}
 \log \mu_j := - \frac{\nu_j}{\lambda p(y_{rj})} \quad \Leftrightarrow \quad \nu_j = - \lambda p(y_{rj}) \log \mu_j \nonumber
\end{align}
it follows that 
$
p_{ij}  = \frac{p(\hat y_{ri})}{\mu_j}  \exp\left( \frac{\frac{\partial J}{\partial p_{ij}}}{\lambda p(y_{rj})} \right) .
$
But since 
$\sum_k p_{kj} \stackrel{!}{=} 1$ for all $j$, we obtain 
$\mu_j = \sum_k p(\hat y_{rk}) \exp \left( \frac{\frac{\partial J}{\partial p_{kj}} }{\lambda p(y_{rj})} \right), \forall~j.$
\\
The optimality conditions are thus given by
\begin{align}
 p_{ij} = p(\hat y_{ri}|y_{rj}) = &\frac{p(\hat y_{ri}) \exp\left( \frac{\frac{\partial J}{\partial p_{ij}}}{\lambda p(y_{rj})} \right)  }{\sum_k p(\hat y_{rk}) \exp \left( \frac{\frac{\partial J}{\partial p_{kj}} }{\lambda p(y_{rj})} \right)}
, &\forall~i,j, \label{eq:optimality_cond_qij}\\
 p(\hat y_{ri}) = & \sum_j p(\hat y_{ri}|y_{rj}) p(y_{rj}),&\forall~i. \label{eq:optimality_marginal}
\end{align}

One can solve for a conditional pmf $p(\hy|y_r)$ satisfying Eqs.~(\ref{eq:optimality_cond_qij}, \ref{eq:optimality_marginal}) with a fixed-point-iteration with an initial value for $p$, similar to the Blahut-Arimoto algorithm (e.g. \cite[Section 10.8]{cover2006elements}).
Different initial values for $p$ can result in different outcomes. 
In practice, we start the iteration with different initial values and store the best result.

\subsection{Scalar vs. Vector quantization}
In general, a vector quantizer must be used at the relay to achieve the rate regions in Sec.~\ref{sec:theory}.
An ideal vector quantizer results in a pmf $p(\hy|y_r)$ that optimizes the problem in Eq.~(\ref{eq:def_I(C)}).
A scalar quantizer permits only  deterministic single-letter relationships, i.e. we have
\begin{align}
\label{eq:q_criteria} 
p(\hat y_{ri}|y_{rj}) = 1~\text{ for some}~ i,~ \forall~j.
\end{align}
If (\ref{eq:q_criteria}) is fulfilled for the optimal pmf, the quantizer function $\mc Q(.)$ can be directly inferred.  
In the saturation region in the curves in Fig.~\ref{fig:I(C)_matlabexample} one obtains scalar quantizers since $C > \log(L) \geq I(Y_r;\hY)$. 
In this case the constraints for problem (\ref{eq:def_I(C)}) are only those for a valid pmf. As a convex maximization is optimized at one of its corner points \cite[Cor. 32.3.2]{rockafellar1997convex}, (\ref{eq:q_criteria}) will be fulfilled.
We see that the loss by using a scalar quantizer is small for sufficiently large $L$ and proceed with this more practical method.

\subsection{Optimized Time Allocation}
$I(C)$ in Eq.~(\ref{eq:def_I(C)}) captures the optimization of the pmf $p(\hy|y_r)$ in problem (\ref{eq:sum_rate_opt_noSI}).
To optimize also the time allocation parameter $\alpha$, we must solve the problem in Eq.~(\ref{eq:I_dependent_on_alpha}).
\begin{prop}
 The function $\alpha I\left(\frac{1-\alpha}{\alpha} C\right)$ is concave in $\alpha$, where $C$ is a positive constant independent of $\alpha$.
\label{prop:concave_alpha}
\end{prop}
\begin{proof}
 From property \ref{property:concave}) in Sec.~\ref{sec:quant_design_sum_rate} we known that 
 $I^{'}(C):=\frac{\partial I(C)}{\partial C} = \lambda \geq 0$ and 
 $I^{''}(C):=\frac{\partial^2 I(C)}{\partial C^2} \leq 0$.
Define the function $h(\alpha) := \frac{1-\alpha}{\alpha} C$ for $0<\alpha\leq 1$.
Note that
\begin{align}
 h^{'}(\alpha) := &\frac{\partial h(\alpha)}{\partial \alpha} = - \frac{1}{\alpha^2} C < 0  ~\text{and} \nonumber\\
 h^{''}(\alpha) := &\frac{\partial^2 h(\alpha)}{\partial \alpha^2} = \frac{2}{\alpha^3} C > 0 .\nonumber
\end{align}
Further note that $h^{'}(\alpha) = -\frac{\alpha}{2} h^{''}(\alpha)$
and
\begin{align}
 \frac{\partial^2 \left(\alpha I(h(\alpha))\right)}{\partial \alpha^2} = & 2 I^{'}(h(\alpha))\cdot h^{'}(\alpha) + \alpha I^{'}(h(\alpha)) \cdot h^{''}(\alpha) +  \nonumber \\
&\alpha I^{''}(h(\alpha))\cdot h^{'}(\alpha)^2.
\end{align}
As $h^{'}(\alpha) = -\frac{\alpha}{2} h^{''}(\alpha)$, the first two summands add to zero and only the last summand remains. This term is always at most $0$, proving that $\frac{\partial^2 \left(\alpha I(h(\alpha))\right)}{\partial \alpha^2} \leq 0$ for $0< \alpha \leq 1$.
\end{proof}
Prop.~\ref{prop:concave_alpha} shows that there is a unique maximizer $\alpha$ for the optimal sum-rate that can be found efficiently once $I(C)$ or an approximation of it is known.

\section{Performance Evaluation}
\label{sec:sims}
Using the method of the previous section, we obtain a mapping for the following parameters: $P= 0$dB, $P_r = 9.3$dB, $N_r = N = 0$dB, $L = 8$, $\alpha = 1/3$. The resulting quantizer fulfills the criteria in (\ref{eq:q_criteria}) (and is thus a scalar quantizer) and gives the sum-rate 0.29. This mapping is shown in Fig. \ref{fig:map3} together with $p(y_\text{r})$.
\begin{figure}[t]
\vspace*{-3mm}
\centering
\includegraphics[width=0.6\columnwidth]{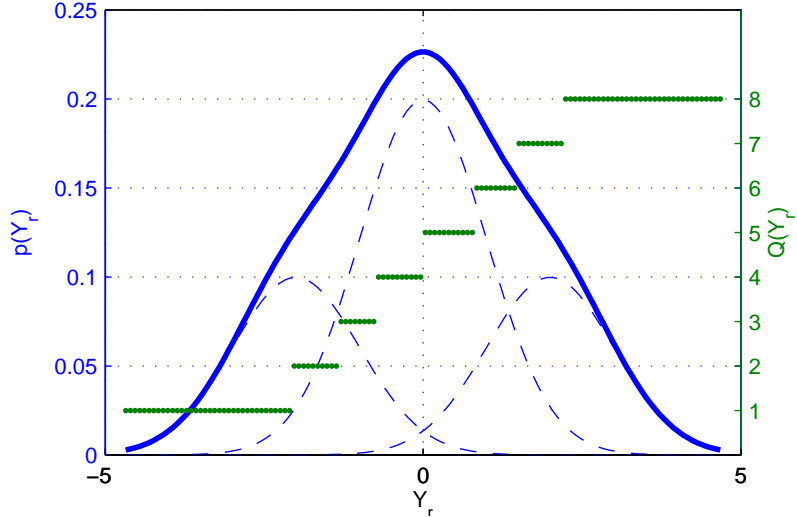}
\caption{Numerically optimized quantizer function (dotted line) and $p(y_\text{r})$ (solid line). Note that the solid curve represents the sum of the dashed curves which correspond to the pdfs conditioned on $X_{1} + X_{2}$. 
}
\label{fig:map3}
\end{figure}

We evaluate the performance of this mapping by means of numerical simulations.
As channel codes, we use the IRA-LDPC codes of the DVB-S2 standard \cite{Morello2006}. 
In the uplink, blocks of $k_1 = k_2 = 7200$ information bits are encoded with the rate 0.444 (MAC Code) to  blocks of $n_\text{MAC} = 16200$ BPSK symbols that are transmitted from \tone\, and \ttwo\,  to \tr\, during the MAC phase. The received 16200 samples at \tr\, are mapped according to the quantizer to 16200 symbols, that are represented by $16200\cdot\log(L)= 48600$ bits\footnote{In general, $\hat{Y}_r$ is not uniformly distributed and hence source coding should be performed before transmitting the indices. However, in this specific example  $H(\hat{Y}_r)\approx 3$ and therefore source coding can be omitted. }.
These bits are encoded by a channel code (BC Code) of rate $3/4$ to 64800 downlink code bits that are broadcast to \tone\, and \ttwo\, during the BC phase with $n_\text{BC} = 32400$ 4-PAM symbols. 
As a result, for the transmission of one block $n_\text{MAC} = 16200$ symbols are used in the uplink and $n_\text{BC} = 32400$ symbols are used in the downlink, which corresponds to $\alpha = 1/3$. 
The sum-rate of the system ($(k_1 + k_2)/(n_\text{MAC}+n_\text{BC})=0.29$ bits/channel use), the time sharing parameter and the transmit powers of the nodes match the optimization parameters of the quantizer. 
At the receivers, we use  the approach given in Sec.~\ref{sec:scheme_NoSW} for decoding: first, the relay quantization index $\hat{q}$ is decoded without using side information. By using $\hat{q}$ and transmitted own symbol, the LLR of the other users symbol is calculated which is fed to the MAC decoder.

During the simulation, the noise level is varied and the packet error rates are evaluated accordingly. 
Fig. \ref{fig:simper} depicts the PER vs. uplink SNR. Recall that the noise powers used in the optimization for this example are chosen as 0dB, which corresponds to an uplink SNR of 0dB. Therefore the PER is expected to approach zero at 0 dB. The gap to this theoretical limit is about 0.75dB. As the DVB-S2 LDPC codes are about 0.7-1.2dB away from the Shannon limit in point-to-point channels\cite{Morello2006} the simulations verify our quantizer design.
\begin{figure}[t]
\vspace*{-3mm}
\centering
\includegraphics[width=0.6\columnwidth]{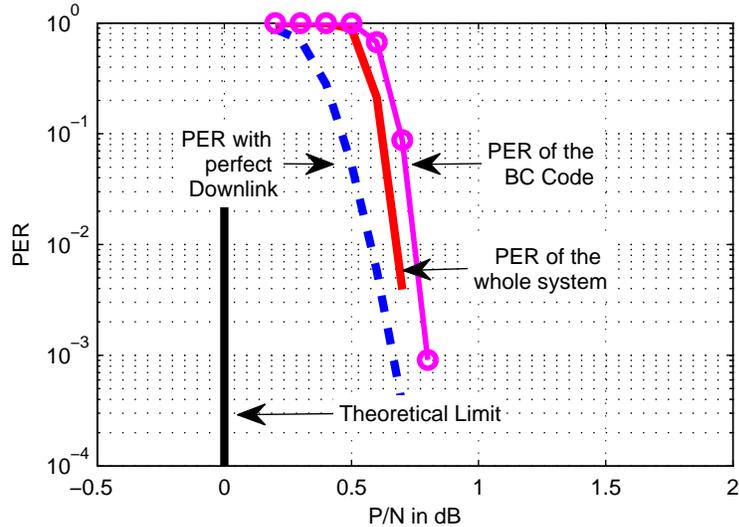}
\caption{Packet Error Rate (PER) simulations for the system with the quantizer given in Fig. \ref{fig:map3}.}
\label{fig:simper}
\vspace*{-2mm}
\end{figure}
Fig. \ref{fig:simper} also shows the PER performance of the BC code (curve with circle marker) and the PER of the system with perfect downlink channels (dashed line). These curves give insight: The gap between the dashed line and the solid line can be seen as the loss due to the imperfect BC code and the gap between the theoretical limit and the dashed line can be interpreted as the loss due to the imperfect MAC code and the quantizer loss. Moreover, it is interesting that the PER of the complete system is less than the PER of BC code. This is because even if some of the quantization indices are transmitted erroneously during the BC phase, they are corrected by the MAC code.

\section{Conclusion}
\label{sec:conclusion}
We compared different QF schemes for the separated TWRC. We showed that the loss caused by neglecting Slepian-Wolf coding in the broadcast code can be small for symmetric setups, but it allows less complex operations at the relay and decoder. A numerical method to maximize the sum-rate and guide quantizer design was derived and applied to special parameters.
We observed that the loss due to scalar instead of vector quantization is small.
Simulations support our results and show that practical schemes are close to the predicted limits.
The study asymmetric scenarios is left as future work.

\section*{Acknowledgments}
The authors are supported by the German Ministry of Education and Research in the framework of the Alexander von Humboldt-Professorship, by the grant DLR@Uni of the Helmholtz Allianz, and by the TUM Graduate School. The authors thank Prof. Gerhard Kramer for his helpful comments.

\bibliographystyle{IEEEtran}

\bibliography{reference}

\begin{thebibliography}{10}
\providecommand{\url}[1]{#1}
\csname url@samestyle\endcsname
\providecommand{\newblock}{\relax}
\providecommand{\bibinfo}[2]{#2}
\providecommand{\BIBentrySTDinterwordspacing}{\spaceskip=0pt\relax}
\providecommand{\BIBentryALTinterwordstretchfactor}{4}
\providecommand{\BIBentryALTinterwordspacing}{\spaceskip=\fontdimen2\font plus
\BIBentryALTinterwordstretchfactor\fontdimen3\font minus
  \fontdimen4\font\relax}
\providecommand{\BIBforeignlanguage}[2]{{%
\expandafter\ifx\csname l@#1\endcsname\relax
\typeout{** WARNING: IEEEtran.bst: No hyphenation pattern has been}%
\typeout{** loaded for the language `#1'. Using the pattern for}%
\typeout{** the default language instead.}%
\else
\language=\csname l@#1\endcsname
\fi
#2}}
\providecommand{\BIBdecl}{\relax}
\BIBdecl

\bibitem{gunduz2008rate}
D.~Gunduz, E.~Tuncel, and J.~Nayak, ``Rate regions for the separated two-way
  relay channel,'' in \emph{Allerton Conf. Communication, Control and
  Computing}.\hskip 1em plus 0.5em minus 0.4em\relax IEEE, 2008, pp.
  1333--1340.

\bibitem{rankov2006achievable}
B.~Rankov and A.~Wittneben, ``Achievable rate regions for the two-way relay
  channel,'' in \emph{IEEE Int. Symp. Inf. Theory}.\hskip 1em plus 0.5em minus
  0.4em\relax IEEE, 2006, pp. 1668--1672.

\bibitem{winkelbauer2011soft}
A.~Winkelbauer and G.~Matz, ``Soft-information-based joint network-channel
  coding for the two-way relay channel,'' in \emph{Int. Symp. Network
  Coding}.\hskip 1em plus 0.5em minus 0.4em\relax IEEE, 2011, pp. 1--5.

\bibitem{zeitler2012}
G.~Zeitler, ``Low-precision quantizer design for communication problems,''
  Ph.D. dissertation, Technische Universit\"at M\"unchen, 2012.

\bibitem{schnurr2008coding}
C.~Schnurr, S.~Stanczak, and T.~Oechtering, ``Coding theorems for the
  restricted half-duplex two-way relay channel with joint decoding,'' in
  \emph{IEEE Int. Symp. Inf. Theory}.\hskip 1em plus 0.5em minus 0.4em\relax
  IEEE, 2008, pp. 2688--2692.

\bibitem{lim2011noisy}
S.~Lim, Y.~Kim, A.~El~Gamal, and S.~Chung, ``Noisy network coding,'' \emph{IEEE
  Trans. Inf. Theory}, vol.~57, no.~5, pp. 3132--3152, 2011.

\bibitem{schnurr2007achievable}
C.~Schnurr, T.~Oechtering, and S.~Stanczak, ``Achievable rates for the
  restricted half-duplex two-way relay channel,'' in \emph{Asilomar Conf.
  Signals, Systems and Computers, 2007}.\hskip 1em plus 0.5em minus 0.4em\relax
  IEEE, 2007, pp. 1468--1472.

\bibitem{kim2011achievable}
S.~J. Kim, N.~Devroye, P.~Mitran, and V.~Tarokh, ``Achievable rate regions and
  performance comparison of half duplex bi-directional relaying protocols,''
  \emph{IEEE Trans. Inf. Theory}, vol.~57, no.~10, pp. 6405 --6418, oct. 2011.

\bibitem{tuncel2006slepian}
E.~Tuncel, ``Slepian-wolf coding over broadcast channels,'' \emph{IEEE Trans.
  Inf. Theory}, vol.~52, no.~4, pp. 1469--1482, 2006.

\bibitem{slepian1973noiseless}
D.~Slepian and J.~Wolf, ``Noiseless coding of correlated information sources,''
  \emph{IEEE Trans. Inf. Theory}, vol.~19, no.~4, pp. 471--480, 1973.

\bibitem{cover2006elements}
T.~Cover and J.~Thomas, \emph{Elements of information theory}.\hskip 1em plus
  0.5em minus 0.4em\relax Wiley-interscience, 2006.

\bibitem{rockafellar1997convex}
R.~Rockafellar, \emph{Convex analysis}.\hskip 1em plus 0.5em minus 0.4em\relax
  Princeton University Press, 1997, vol.~28.

\bibitem{witsenhausen1975conditional}
H.~Witsenhausen and A.~Wyner, ``A conditional entropy bound for a pair of
  discrete random variables,'' \emph{IEEE Trans. Inf. Theory}, vol.~21, no.~5,
  pp. 493--501, 1975.

\bibitem{Morello2006}
A.~Morello and V.~Mignone, ``{DVB-S2: The Second Generation Standard for
  Satellite Broad-Band Services},'' \emph{Proc. IEEE}, vol.~94, no.~1, pp.
  210--227, 2006.

\bibitem{kramer2007topics}
G.~Kramer, ``Topics in multi-user information theory,'' \emph{Found. and Trends
  in Comm. and Inf. Theory}, vol.~4, no. 4-5, pp. 265--444, 2007.

\bibitem{el2010lecture}
A.~El~Gamal and Y.~Kim, ``Lecture notes on network information theory,''
  \emph{CoRR}, vol. abs/1001.3404, 2010.

\bibitem{schnurr2008thesis}
C.~Schnurr, ``Achievable rates and coding strategies for the two-way relay
  channel,'' Ph.D. dissertation, Technische Universit\"at Berlin, 2008.

\bibitem{yeung2008information}
R.~Yeung, \emph{Information theory and network coding}.\hskip 1em plus 0.5em
  minus 0.4em\relax Springer Verlag, 2008.

\bibitem{boyd2004convex}
S.~Boyd and L.~Vandenberghe, \emph{Convex optimization}.\hskip 1em plus 0.5em
  minus 0.4em\relax Cambridge University Press, 2004.

\end{thebibliography}

\clearpage


\newcommand{\yhcard}{\ensuremath{L}}

\appendix

\subsection{Achievability Proof}
\label{sec:achievability}

Before we proove the rate region in (\ref{eq:RRnoSI}), we use the following usual definition of typical sequences. Notation and definitions essentially follow \cite{kramer2007topics,el2010lecture}. The analysis follows the style of \cite{schnurr2007achievable,schnurr2008coding}. 

Let $x^n$ be a sequence with each element $x_i$ drawn from a finite alphabet $\mc X$. The number of elements in $x^n$ taking the letter $a \in \mc X$ is denoted by $N(a|x^n)$.

Define the typical set $\mc T_\epsilon^n(X)$ as all sequences $x^n$ satisfying
\begin{equation}
 \left| \frac{1}{n} N(a|x^n) - P_X(a)  \right| \leq \epsilon \cdot P_X(a) \qquad \forall a \in \mc X
\end{equation}
A sequence $x^n \in \mc T_\epsilon^n(X)$ is called \emph{$\epsilon$-letter-typical} or just \emph{typical} with respect to $P_X(\cdot)$. Similarly for joint distributions and joint typicality.

\subsubsection{Random Codebook Generation}

Define $n_{\text{MAC}} = \alpha n>0$, $n_{\text{BC}} = (1-\alpha) n>0$.

\begin{itemize}
 \item Choose $u^{n_{\text{MAC}}}$ according to $\prod_{k=1}^{n_{\text{MAC}}} P_U(u_i)$.
 \item Randomly and independently generate $2^{nR_1}$ codewords according to $\prod_{k=1}^{n_{\text{MAC}}} P_{X_1|U}(x_{1k}|u_k)$. Label them $x_1^{n_{\text{MAC}}}(w_1)$ with $w_1 \in \{1,2,\ldots , 2^{nR_1}\}$. 
 \item Randomly and independently generate  $2^{nR_2}$ codewords according to $\prod_{k=1}^{n_{\text{MAC}}} P_{X_2|U}(x_{2k}|u_k)$. Label them $x_2^{n_{\text{MAC}}}(w_2)$ with $w_2 \in \{1,2,\ldots , 2^{nR_2}\}$.
 \item Randomly and independently generate  $2^{n_{\text{MAC}} R_Q}$ codewords according to $\prod_{k=1}^{n_{\text{MAC}}} P_{\hY|U}(\hat{y}_{rk}|u_k)$. Label them $\hat y_{r}^{n_{\text{MAC}}}(i)$ with $i \in \{1,2,\ldots , 2^{n_{\text{MAC}} R_Q}\}$.
 \item Randomly and independently generate  $2^{n_{\text{MAC}} R_Q}$ codewords  according to $\prod_{k=1}^{n_{\text{BC}}} P_{X_r}(x_{rk})$. Label them $x_{r}^{n_{\text{BC}}}(i)$ with $i \in \{1,2,\ldots , 2^{n_{\text{MAC}} R_Q}\}$.
\end{itemize}
Reveal all these codebooks to all nodes.

\subsubsection{Coding}
\begin{enumerate}
 \item In order to transmit message $w_1$, node 1 sends $x_1^{n_{\text{MAC}}}(w_1)$.
 \item In order to transmit message $w_2$, node 2 sends $x_2^{n_{\text{MAC}}}(w_2)$.
 \item When the relay receives $y_r^{n_{\text{MAC}}}$, it looks for the first index $i$ such that $(y_r^{n_{\text{MAC}}}, \hat y_r^{n_{\text{MAC}}}(i)) \in \mc T_{\epsilon_1}^{n_{\text{MAC}}} (Y_r, \hY | u^{n_{\text{MAC}}})$. If it does not find such an index, $i=1$.
The relay transmits $x_r^{n_{\text{BC}}}(i)$.
 \item When node 1 receives $y_1^{n_{\text{BC}}}$, it looks for the unique $i$ such that $(x_r^{n_{\text{BC}}}(i),y_1^{n_{\text{BC}}}) \in \mc T_{\epsilon_2}^{n_{\text{BC}}} (X_r, Y_1)$. If no or more than one such $i$ is found, node 1 chooses $\hat w_2 = 1$. Node 1 thus knows $\hat y_r^{n_{\text{MAC}}}(i)$.
 \item Node 1 decides for the unique $\hat w_2$ that satisfies $(x_1^{n_{\text{MAC}}}(w_1), x_2^{n_{\text{MAC}}}(\hat w_2), \hat y_r^{n_{\text{MAC}}}(i)) \in \mc T_{\epsilon_3}^{n_{\text{MAC}}}(X_1, X_2,\hY|u^{n_{\text{MAC}}})$. If no or more than one such index is found, node 1 sets $\hat w_2 = 1$. 
 \item Steps 4 and 5 similarly for node 2.
\end{enumerate}

\subsubsection{Error events}

\begin{enumerate}
 \item[$\mc E_1$] is the event that coding step 3 fails. 
\\That is, $\nexists i \in \{1,2,\ldots , 2^{nR_Q}\}$ such that $(y_r^{n_{\text{MAC}}}, \hat y_r^{n_{\text{MAC}}}(i)) \in \mc T_{\epsilon_1}^{n_{\text{MAC}}} (Y_r, \hY | u^{n_{\text{MAC}}})$. This is not an intrinsic error event but simplifies the analysis.

 \item[$\mc E_2$] is the event that coding step 4 fails. 
\\That is,  $(x_r^{n_{\text{BC}}}(i),y_1^{n_{\text{BC}}}) \not \in \mc T_{\epsilon_2}^{n_{\text{BC}}} (X_r, Y_1)$ or  $(x_r^{n_{\text{BC}}}(j),y_1^{n_{\text{BC}}}) \not \in \mc T_{\epsilon_2}^{n_{\text{BC}}} (X_r, Y_1)$ with $i \not = j$.

 \item[$\mc E_3$] is the event that coding step 5 fails. 
\\Assume $x_1^{n_{\text{MAC}}}(w_1)$ and $x_2^{n_{\text{MAC}}}(w_2)$ are sent and the relay chooses the index $i$ such that $(y_r^{n_{\text{MAC}}}, \hat y_r^{n_{\text{MAC}}}(i)) \in \mc T_{\epsilon_1}^{n_{\text{MAC}}} (Y_r, \hY | u^{n_{\text{MAC}}})$. 
$\mc E_3$ denotes the event that $(x_1^{n_{\text{MAC}}}(w_1), x_2^{n_{\text{MAC}}}(w_2), \hat y_r^{n_{\text{MAC}}}(i) ) \not \in \mc T_{\epsilon_3}^{n_{\text{MAC}}}(X_1, X_2,\hY|u^{n_{\text{MAC}}})$ or $(x_1^{n_{\text{MAC}}}(w_1), x_2^{n_{\text{MAC}}}(\hat w_2), \hat y_r^{n_{\text{MAC}}}(i) ) \in \mc T_{\epsilon_3}^{n_{\text{MAC}}}(X_1, X_2,\hY|u^{n_{\text{MAC}}})$ with $\hat w_2 \not = w_2$.
\end{enumerate}

The error events and achievable rates for node 1 are similar.

It is not hard to check that the probability of error $P(\mc E)$ is bounded by the probability of the union of these error events, i.e.
$$ P(\mc E) \leq P\left( \bigcup_{k=1}^3 \mc E_k \right) \leq \sum_{k=1}^3 \mc E_k.$$

\subsubsection{Error analysis}

\underline{Event $\mc E_1$:}

Joint typicality of $(y_r^{n_{\text{MAC}}}, \hat y_r^{n_{\text{MAC}}}(i)) $ requires $y_r^{n_{\text{MAC}}} \in \mc T_{\epsilon_{1'}}(Y_r|u^{n_{\text{MAC}}})$. However, $P\{ y_r^{n_{\text{MAC}}} \not \in \mc T_{\epsilon_{1'}}(Y_r|u^{n_{\text{MAC}}}) \} \rightarrow 0$ for $n_{\text{MAC}}~\rightarrow~\infty$ .

As the codewords $\hat y_r^{n_{\text{MAC}}}(i)$, $i \in \{1,2,\ldots , 2^{n_{\text{MAC}} R_Q}\}$ are drawn i.i.d., 
\begin{align}
 P(\mc E_1) = &P \{ \nexists i: (y_r^{n_{\text{MAC}}}, \hat y_r^{n_{\text{MAC}}}(i)) \in \mc T_{\epsilon_1}^{n_{\text{MAC}}} (Y_r, \hY | u^{n_{\text{MAC}}}) \} = \nonumber 
\\
 &=\sum_{y_r^{n_{\text{MAC}}} \in \mc T_{\epsilon_{1'}^{n_{\text{MAC}}},}(Y_r|u^{n_{\text{MAC}}})} P_{Y_r^{n_{\text{MAC}}}|U^{n_{\text{MAC}}}}(y_r^{n_{\text{MAC}}}|u^{n_{\text{MAC}}}) \left[ 1 - P\{ (y_r^{n_{\text{MAC}}}, \hat Y_r^{n_{\text{MAC}}} ) \in \mc T_{\epsilon_1}^{n_{\text{MAC}}}(Y_r,\hY|u^{n_{\text{MAC}}} ) \} \right]^{2^{n_{\text{MAC}} R_Q}} 
\\
&\stackrel{(a)}{\leq} \sum_{y_r^{n_{\text{MAC}}} \in \mc T_{\epsilon_{1'}^{n_{\text{MAC}}},}(Y_r|u^{n_{\text{MAC}}})} P_{Y_r^{n_{\text{MAC}}}|U^{n_{\text{MAC}}}}(y_r^{n_{\text{MAC}}}|u^{n_{\text{MAC}}}) \left[ 1 - 2^{-n_{\text{MAC}} (I(Y_r;\hY|u^{n_{\text{MAC}}}) + \delta(\epsilon_1))} \right]^{2^{n_{\text{MAC}} R_Q}}
\\
&\stackrel{(b)}{\leq} \sum_{y_r^{n_{\text{MAC}}} \in \mc T_{\epsilon_{1'}^{n_{\text{MAC}}},}(Y_r|u^{n_{\text{MAC}}})} P_{Y_r^{n_{\text{MAC}}}|U^{n_{\text{MAC}}}}(y_r^{n_{\text{MAC}}}|u^{n_{\text{MAC}}}) \exp\left(- 2^{n_{\text{MAC}} (R_Q - (I(Y_r;\hY|U) + \delta(\epsilon_1)))} \right)
\\
&\stackrel{(c)}{\leq} \exp\left(- 2^{n_{\text{MAC}} (R_Q - (I(Y_r;\hY|U) + \delta(\epsilon_1)))} \right) = \exp\left(- 2^{n \alpha (R_Q - (I(Y_r;\hY|U) + \delta(\epsilon_1)))} \right)
\end{align}

where 
\\$(a)$ follows from the joint typicality lemma \cite{el2010lecture} bounding
the probability that a given $y_r^{n_{\text{MAC}}}$ is jointly typical with a randomly independently sampled codeword $\hat y_r^{n_{\text{MAC}}}(i)$ by
\begin{align}
 P\{ (y_r^{n_{\text{MAC}}}, \hat Y_r^{n_{\text{MAC}}} ) \in \mc T_{\epsilon_1}^{n_{\text{MAC}}}(Y_r,\hY|u^{n_{\text{MAC}}} ) \} \geq 2^{-n_{\text{MAC}} (I(Y_r;\hY|U) + \delta(\epsilon_1))}.
\end{align}
$(b)$ follows from $(1-x)^n \leq \exp(-n\cdot x)$ \cite{cover2006elements}.\\
$(c)$ is due to $\sum_{y_r^{n_{\text{MAC}}} \in \mc T_{\epsilon_{1'}}(Y_r|u^{n_{\text{MAC}}})} P_{Y_r^{n_{\text{MAC}}}|U^{n_{\text{MAC}}}}(y_r^{n_{\text{MAC}}}|u^{n_{\text{MAC}}}) < 1$

In the limit it follows that $P(\mc E_1) \rightarrow 0$ for $n \rightarrow \infty$ for $\epsilon_1 > \epsilon_{1'} \geq 0$ if
\begin{equation}
 \boxed{ \alpha R_Q > \alpha I(Y_r;\hY|U) .}
\end{equation}

\underline{Event $\mc E_2$:}

This is the classical proof of the channel coding theorem. We can split this event in  two subevents $\mc E_{21}$ and $\mc E_{22}$ with $P(\mc E_2) \leq P(\mc E_{21}) + P(\mc E_{22})$.
Define 
\begin{align}
\mc E_{21} :=& \{ (x_r^{n_{\text{BC}}}(i),y_1^{n_{\text{BC}}}) \not \in \mc T_{\epsilon_2}^{n_{\text{BC}}} (X_r, Y_1) \} \nonumber\\
\mc E_{22} :=&\{ (x_r^{n_{\text{BC}}}(j), y_1^{n_{\text{BC}}}) \in \mc T_{\epsilon_2}^{n_{\text{BC}}}(X_r, Y_1) \} \text{ for some } j \not = i \nonumber
\end{align}

Since $(X_r^{n_{\text{BC}}}(i), Y_1^{n_{\text{BC}}}) \sim \prod_{k=1}^{n_{\text{BC}}} P_{X_r,Y_1|U}(x_{rk}, y_{1k}|u_k) $, by the law of large numbers
\begin{align}
 P(\mc E_{21}) \rightarrow 0 \text{ for } n_{\text{BC}} \rightarrow \infty .
\end{align}

For $\mc E_{22}$, by symmetry, we focus on the case that $x_r^{n_{\text{BC}}}(1)$ was transmitted. 
\\Note that $(X_r^{n_{\text{BC}}}(j), Y_1^{n_{\text{BC}}}) \sim \prod_{k=1}^{n_{\text{BC}}} P_{X_r|U}(x_{rk}|u_k) P_{Y_1|U}(y_{1k}|u_k) $ for $i \not = j$.

\begin{align}
 P(\mc E_{22}) &\leq \sum_{j=2}^{2^{n_{\text{MAC}} R_Q}} P \{ (x_r^{n_{\text{BC}}}(j), y_1^{n_{\text{BC}}}) \in \mc T_{\epsilon_2}^{n_{\text{BC}}}(X_r, Y_1) \} \nonumber 
\\
&\stackrel{(a)}{\leq} 2^{n_{\text{MAC}} R_Q} 2^{-n_{\text{BC}} (I(X_r;Y_1) - \delta(\epsilon_2) )} = 2^{-n (1-\alpha (I(X_r;Y_1) -\delta(\epsilon_2)) -\alpha R_Q )}
\end{align}

where
\\$(a)$ follows from the joint typicality lemma stating
$$ P \{ (x_r^{n_{\text{BC}}}(j), y_1^{n_{\text{BC}}}) \in \mc T_{\epsilon_2}^{n_{\text{BC}}}(X_r, Y_1) \} \leq 2^{-n_{\text{BC}} (I(X_r;Y_1) - \delta(\epsilon_2))}.$$
Concludingly, in the limit $P(\mc E_2) \rightarrow 0$ for $n \rightarrow \infty$ for $\epsilon_2 > 0$ if
\begin{equation}
 \boxed{\alpha R_Q < 1-\alpha I(X_r;Y_1) .}
\end{equation}

\underline{Event $\mc E_3$:}

Again, we can split this into two subevents $\mc E_{31}$ and $\mc E_{32}$ with $P(\mc E_3) \leq P(\mc E_{31}) + P(\mc E_{32})$.
Define
\begin{align}
\mc E_{31} :=& \{ (x_1^{n_{\text{MAC}}}(w_1), x_2^{n_{\text{MAC}}}(w_2), \hat y_r^{n_{\text{MAC}}}(i)) \not \in \mc T_{\epsilon_3}^{n_{\text{MAC}}}(X_1,X_2,\hY|u^{n_{\text{MAC}}}) \} \nonumber\\
\mc E_{32} :=&\{ (x_1^{n_{\text{MAC}}}(w_1), x_2^{n_{\text{MAC}}}(\hat w_2), \hat y_r^{n_{\text{MAC}}}(i))  \in \mc T_{\epsilon_3}^{n_{\text{MAC}}}(X_1,X_2,\hY|u^{n_{\text{MAC}}}) \} \text{ for some } \hat w_2 \not = w_2 \nonumber
\end{align}

From the coding scheme, $\hY^{n_{\text{MAC}}} = \mc Q(Y_r^{n_{\text{MAC}}})$ and thus $(X_1,X_2) \rightarrow Y_r \rightarrow \hY$.

For $\mc E_{31}$, by the law of large numbers, $P\{ (x_1^{n_{\text{MAC}}}(w_1), x_2^{n_{\text{MAC}}}(w_2), y_r^{n_{\text{MAC}}}) \in \mc T_{\epsilon_{3'}}^{n_{\text{MAC}}}(X_1,X_2,Y_r|u^{n_{\text{MAC}}}) \} \rightarrow 1$ for $n_{\text{MAC}} \rightarrow \infty$.
By the Markov Lemma \cite{el2010lecture}, typicality of $(x_1^{n_{\text{MAC}}}(w_1), x_2^{n_{\text{MAC}}}(w_2), y_r^{n_{\text{MAC}}})$ implies
\begin{align}
 P\{ (x_1^{n_{\text{MAC}}}(w_1), x_2^{n_{\text{MAC}}}(w_2), y_r^{n_{\text{MAC}}}, \hat y_r^{n_{\text{MAC}}}(i)) \in \mc T_{\epsilon_{3^*}}^{n_{\text{MAC}}}(X_1,X_2,Y_r,\hY|u^{n_{\text{MAC}}}) \} \rightarrow 1 \text{ for } n_{\text{MAC}} \rightarrow \infty
\end{align}
for $0 \leq \epsilon_{3'}  < \epsilon_{3^*}$.
A direct consequence is that 
\begin{align}
 P\{ (x_1^{n_{\text{MAC}}}(w_1), x_2^{n_{\text{MAC}}}(w_2), \hat y_r^{n_{\text{MAC}}}(i)) \not \in \mc T_{\epsilon_3}^{n_{\text{MAC}}}(X_1,X_2,\hY|u^{n_{\text{MAC}}}) \} \rightarrow 0 \text{ for } n_{\text{MAC}} \rightarrow \infty.
\end{align}
for $0 \leq \epsilon_3 < \epsilon_{3^*}$.

For $\mc E_{32}$:

This part is similar to the multi-access channel.
\\Note that $(X_1^{n_{\text{MAC}}}(w_1), X_2^{n_{\text{MAC}}}(\hat w_2), \hY^{n_{\text{MAC}}}(i)) \sim \prod_{k=1}^{n_{\text{MAC}}} P_{X_1,\hY|U}(x_{1k},\hat y_{rk}|u_k) P_{X_2|U}(x_{2k}|u_k)$ for $\hat w_2 \not = w_2$.
By introducing the random variable $V^{n_{\text{MAC}}} = (X_1^{n_{\text{MAC}}},\hY^{n_{\text{MAC}}})$ we can use the same steps as for error event $\mc E_{22}$ and conclude that

\begin{align}
 P(\mc E_{32}) &\leq \sum_{j=2}^{2^{n R_2}} P \{ (x_2^{n_{\text{MAC}}}(\hat w_2), v^{n_{\text{MAC}}}) \in \mc T_{\epsilon_3}^{n_{\text{MAC}}}(X_2, V|u^{n_{\text{MAC}}}) \} \nonumber 
\\
&\stackrel{(a)}{\leq} 2^{n R_2} 2^{-n_{\text{BC}} (I(X_2;V|U) - \delta(\epsilon_3) )} = 2^{-n ( 1-\alpha (I(X_2;V|U) -\delta(\epsilon_3)) - R_2 )}
\end{align}

where
\\$(a)$ follows from the joint typicality lemma.
With $I(X_2;V|U) = I(X_2; X_1,Y_r|U) = I(X_2;Y_r|X_1,U)$ due to the independence of $X_1$ and $X_2$ given $U$, in the limit $P(\mc E_3)\rightarrow 0$ for $n \rightarrow \infty$ if

\begin{align}
 \boxed{R_2 < \alpha I(X_2;Y_r|X_1,U).}
\end{align}

\subsubsection{Cardinality of $\mc U$}
This derivation closely follows the proof in \cite[Chapter 3.2]{schnurr2008thesis}.

The achievable rate region can be written as
\begin{align}
R_1 &\leq \sum_{u\in \mc U} p(u) \alpha I(X_1;\hY|X_2, U = u) & := \sum_{u\in \mc U} p(u) g_1(p(x_1,x_2|u))\nonumber \\
R_2 &\leq \sum_{u\in \mc U} p(u) \alpha I(X_2;\hY|X_1, U = u) & := \sum_{u\in \mc U} p(u) g_2(p(x_1,x_2|u))\nonumber \\
0 &\leq \sum_{u \in \mc U} p(u) \left[ (1-\alpha) I(X_r;Y_2) - \alpha I(Y_r;\hY|U=u) \right] & := \sum_{u\in \mc U} p(u) g_3(p(x_1,x_2|u))\nonumber \\
0 &\leq \sum_{u \in \mc U} p(u) \left[ (1-\alpha) I(X_r;Y_1) - \alpha I(Y_r;\hY|U=u) \right] & := \sum_{u\in \mc U} p(u) g_4(p(x_1,x_2|u)).
\end{align}

The rate region can be interpreted as a subset of the convex hull of the region in the $4$-dimensional space spanned by the functions $(g_1(p(x_1,x_2|u)), g_2(p(x_1,x_2|u)), g_3(p(x_1,x_2|u)), g_4(p(x_1,x_2|u))$, which only depend on the conditional pmf $p(x_1,x_2|u)$.
\\Let $\mc S$ denote the set of all such points for each choice of the compact set $p(x_1)p(x_2)p(y_r|x_1,x_2)p(\hy|y_r)p(x_r)p(y_1,y_2|x_r)$.
Precisely, 
$$\mc S = \bigcup_{p(x_1)p(x_2)p(y_r|x_1,x_2)p(\hy|y_r)p(x_r)p(y_1,y_2|x_r)} \{g_1(p(x_1,x_2)), g_2(p(x_1,x_2)),g_3(p(x_1,x_2)),g_4(p(x_1,x_2))\}$$
As continuous image of a compact set, $\mc S$ is connected.
Define $\mc C = \text{conv}(\mc S)$: By the Fenchel-Eggleston strenghtening of Carath\'{e}odory's theorem (e.g. \cite[Appendix A, C]{el2010lecture}), every point in $\mc C$ can be obtained by taking a convex combination of at most $4$ points in $\mc S$. As the rate region is a subset of $\mc C$, it follows that $|\mc U| \leq 4$.

\subsubsection{Cardinality of $\mc{\hat Y}_r$}

Similarly, the rate region can be written as

\begin{align}
R_1 &\leq \sum_{\hy\in \mc{\hat Y}_r} p(\hy) \alpha \left[ H(X_1|X_2,U) - H(X_1|X_2, U, \hY = \hy)\right] & := \sum_{\hy\in \mc{\hat Y}_r} p(\hy) g_1(p(y_r|\hy))\nonumber \\
R_2 &\leq \sum_{\hy\in \mc{\hat Y}_r} p(\hy) \alpha \left[H(X_2|X_1,U) - H(X_2|X_1, U, \hY = \hy)\right] & := \sum_{\hy\in \mc{\hat Y}_r} p(\hy) g_2(p(y_r|\hy))\nonumber \\
0 &\leq \sum_{\hy\in \mc{\hat Y}_r} p(\hy) \left[ (1-\alpha) I(X_r;Y_2) - \alpha \left( H(Y_r|U) - H(Y_r|U,\hY = \hy) \right) \right] & := \sum_{\hy\in \mc{\hat Y}_r} p(\hy) g_3(p(y_r|\hy)) \nonumber \\
0 &\leq \sum_{\hy\in \mc{\hat Y}_r} p(\hy) \left[ (1-\alpha) I(X_r;Y_1) - \alpha \left( H(Y_r|U) - H(Y_r|U,\hY = \hy) \right) \right] & := \sum_{\hy\in \mc{\hat Y}_r} p(\hy) g_4(p(y_r|\hy)).
\end{align}
Additionally, the following $|\mc Y_r|-1$ conditions have to be met:
\begin{align}
 p(Y_r = y_{ri}) = \sum_{\hy\in \mc{\hat Y}_r} p(\hy) p(Y_r = y_{ri}|\hY = \hy) & :=  \sum_{\hy\in \mc{\hat Y}_r} p(\hy) g_{i+4}(p(y_r|\hy))
\qquad \forall i = 1, \ldots, |\mc Y_r|-1
\end{align}

Again, the rate region can be interpreted as a subset of the convex hull of the $|\mc Y_r|+3$-dimensional space spanned by the functions $g_1(p(y_r|\hy)), \ldots, g_{|\mc Y_r|+3}(p(y_r|\hy))$, only depending on the conditional pmf $p(y_r|\hy)$.
Precisely, let $\mc S$ denote the set of all points spanned by $g_1(p(y_r|\hy)), \ldots, g_{|\mc Y_r|+3}(p(y_r|\hy))$ for each choice of the compact set $p(y_r|\hy)$, i.e.
$$\mc S = \bigcup_{p(y_r|\hy)} \{g_1(p(y_r|\hy)), \ldots, g_{|\mc Y_r|+3}(p(y_r|\hy))\} $$ 
As continuous image of a compact set, $\mc S$ is connected. Define $\mc C = \text{conv}(\mc S)$: By the Fenchel-Eggleston strenghtening of Carath\'{e}odory's theorem (e.g. \cite[Appendix A, C]{el2010lecture}), every point in $\mc C$ can be obtained by taking a convex combination of at most $|\mc Y_r|+3$ points in $\mc S$. As the rate region is a subset of $\mc C$, it follows that $|\mc{\hat Y}_r| \leq |\mc Y_r|+3$.

\subsection{Derivation of the claims in Sec.~\ref{sec:sum_rate_comparison}}
\label{sec:optimal_hatN}

For $\mc R_\text{CF}$ and $\mc R_\text{NoSW}$ the expression in Eq.~(\ref{eq:Rsum_alpha}) is straight forward.
For $\mc R_\text{NNC}$, the sum rate problem for the symmetric case can be written as
\begin{align}
 R^{\text{sum}}_{\text{NNC}}(\alpha) \max_{\hat N} \min\{ 2 \alpha C\left(  \frac{P}{N_r + \hat N}\right), 2\left( (1-\alpha) C\left(  \frac{P_r}{N} \right) - \alpha C \left(  \frac{N_r}{\hat N} \right) \right)\} \nonumber
\end{align}
The first expression inside the minimum is strictly decreasing an convex in $\hat N$. The second one is strictly increasing and concave in $\hat N$. Implicitly, we always assume that the second expression is $\geq 0$.
The function represented by the minimum is thus quasi-concave in $\hat N$. In particular, the maximal value is attained where both expressions inside the minimum are the same.
For that,
\begin{align}
2 \alpha C\left(  \frac{P}{N_r + \hat N}\right) &= 2\left( (1-\alpha) C\left(  \frac{P_r}{N} \right) - \alpha C \left(  \frac{N_r}{\hat N} \right) \right) \nonumber\\
\alpha \log\left( \frac{N_r + \hat N + P}{N_r + \hat N} \cdot \frac{\hat N + N_r}{\hat N}\right) &= (1-\alpha) \log\left( 1 + \frac{P_r}{N}\right) \nonumber \\
1 + \frac{N_r  + P }{\hat N } &= \left( 1 + \frac{P_r}{N}\right)^{(1-\alpha)/\alpha} \nonumber \\
\Rightarrow \hat N^*_\text{NNC}(\alpha) &= \frac{P + N_r}{(1 + P_r/N)^{(1-\alpha)/\alpha} -1 } 
\end{align}
which equals $\hat N^*_\text{CF}(\alpha)$

\subsection{Derivation of the properties of Sec.~\ref{sec:quant_design_sum_rate}}
\label{sec:derivation_propertiesIC}
We restate the properties with full explanation here

\begin{itemize}
 \item 
$I(C)$ is upper bounded by 
\begin{align}
 I(C) \leq  \left( I(X_1;Y_r|X_2) + I(X_2;Y_r|X_1) \right),
\label{eq:I(C)UB}
\end{align}
with equality if $ C \geq H(Y_r)$.
This is due to $I(Y_r;\hY) \leq H(Y_r)$, with equality if $H(Y_r|\hY)=0$.
Now, the upper bound in Eq.~(\ref{eq:I(C)UB}) is met if $I(X_1;Y_r|\hY,X_2) = 0$, i.e.
\begin{align}
 I(X_1;Y_r|\hY,X_2) = \underbrace{H(Y_r|\hY,X_2)}_{=0} - \underbrace{H(Y_r|\hY,X_1,X_2)}_{=0}
\end{align}

\item The functions $I(X_1;\hY|X_2) + I(X_2; \hY|X_1)$ and $ I(Y_r;\hY)$ are
convex in $p(\hy|y_r)$ \cite[Theorem 2.7.4]{cover2006elements}. Computing $I(C)$ thus requires maximizing a convex function over
a convex set, a difficult problem in general.
From the \emph{maximum principle} \cite[Cor. 32.3.2]{rockafellar1997convex}, it follows that for the optimal $p(\hy|y_r)$, $ I(Y_r;\hY) = C$.

\item $I(C)$ is an increasing and concave function in
$C$, for $0 \leq C \leq H(Y_r)$. 
The first part follows by the consideration before. Thus, for $0\leq C_1 < C_2 \leq  H(Y_r)$,  $I(C_1) < I(C_2)$. As a consequence, one can restrict the optimization to the constraint $ I(Y_r;\hY) =C$.
The second part will be proved in the following.

\end{itemize}
Recall that we abbreviate $p(\hy|y_r)$ by $p$:
We rewrite the problem, similar to \cite{witsenhausen1975conditional} and \cite{zeitler2012}:
\begin{align}
 I(C) = 
&\sup_{p}  \left(H(X_1) + H(X_2) - H(X_1|\hY,X_2) - H(X_2|\hY,X_1) \right), \nonumber \\ 
&\text{s.t. } H(Y_r) - H(Y_r|\hY) =  C, \quad 0\leq C \leq  H(Y_r).
\end{align}
By dropping the constant terms, an equivalent problem (in the sense of the same optimal argument) is given by
\begin{align}
 F(x) := \inf_{p}  \left( H(X_1|\hY,X_2) + H(X_2|\hY,X_1)\right), \nonumber \\
\text{s.t. } H(Y_r|\hY) = x, \quad 0\leq x \leq H(Y_r).
 \label{eq:Def_F(x)}
\end{align}
We investigate properties of $F(x)$.

In the following, we often use a vector representation of marginal probability distributions. 
The distribution of a general random variable $Z$, $p(z)$ is equivalently represented by the column vector $\ve p_{z} \in \Delta_{|\mc{Z}|}$ in the $|\mc{Z}|$-dimensional probability simplex $\Delta_{|\mc{Z}|}$, describing an $(|\mc{Z}|-1)$-dimensional space. 
The $i$-th coordinate is denoted by $p_{z,i} = p(Z=z_i)$.

Therefore, let $\ve{p}_{y_r} \in \Delta_{|\mc{Y}_r|}$, $\ve{p}_{\hat{y}_r} \in \Delta_{{\yhcard}}$ represent the marginal distribution $p(y_r)$,  $p(\hy)$, respectively.

Let $B = [\ve b_1, \ldots , \ve b_{{\yhcard}}]$ be a $|\mc{Y}_r| \times {\yhcard}$ stochastic matrix with $\ve b_{i} \in \Delta_{|\mc{Y}_r|}$ in the $i$-th column.

Introduce the random variables 
\begin{itemize}
 \item $Y_r'$ with marginal distribution  
  $\ve p_{y_r'} = B \ve{p}_{\hat{y}_r} = \sum_{i=1}^{{\yhcard}} p(\hat y_{ri}) \ve b_i$,
 \item $X_1'$ with marginal distribution  
  $p(x_{1j}') =   \sum_{i=1}^{|\mc Y_r|} p(x_{1j}|y_{ri}) p(y_{ri}'),~\forall~j=1,\ldots,|\mc X_1| $,
 \item $X_2'$ with marginal distribution 
  $p(x_{2j}') =   \sum_{i=1}^{|\mc Y_r|} p(x_{2j}|y_{ri}) p(y_{ri}'),~\forall~j=1,\ldots,|\mc X_2| $
\end{itemize}

In general, the matrix $B$ corresponds to $p(y_r'|\hy)$.
Clearly, if $B$ is equal to to $p(y_r|\hy)$, then $p(y_r') = p(y_r)$, $p(x_1') = p(x_1)$ and $p(x_2') = p(x_2)$.

One can write:
\begin{align}
 \ve p_{y_r'} &= \sum_{i=1}^{{\yhcard}} p(\hat y_{ri}) \ve b_i  \label{eq:p_proof}\\
 \xi & = H(Y_r'|\hY) = \sum_{i=1}^{{\yhcard}} p(\hat y_{ri}) H(Y_r'|\hY = \hat y_{ri}) 
  := \sum_{i=1}^{{\yhcard}} p(\hat y_{ri}) h_{|\mc{Y}_r|}(\ve b_i) \leq \log(|\mc{Y}_r|)  \label{eq:xi_proof} \\
 &\text{with $h_{|\mc{Y}_r|}(\ve p) = \sum_{j=1}^{|\mc{Y}_r|} p_j \log(p_j)$ as the entropy function.}\nonumber \\
 \eta & = H(X_1'|X_2',\hY) + H(X_2'|X_1', \hY)  \nonumber \\
      & = \sum_{i=1}^{{\yhcard}} p(\hat y_{ri}) \left[ \underbrace{H(X_1'|X_2', \hY = \hat y_{ri})}_{\leq \log(|\mc{X}_1|)}   \underbrace{H(X_2'|X_1', \hY = \hat y_{ri})}_{\leq \log(|\mc{X}_2|)} \right]  
       := \sum_{i=1}^{{\yhcard}} p(\hat y_{ri}) g(\ve b_i) \leq \log(|\mc{X}_1|\cdot |\mc{X}_2| ) \label{eq:eta_proof}
\end{align}

The problem in Eq.~(\ref{eq:Def_F(x)}) can be stated as
\begin{align}
F(x) =  \inf_{\ve p_{y_r'} = \ve{p}_{y_r}} \eta, \qquad \text{s.t. } \xi = x.
\end{align}

Now define the mapping $\ve b_i \in \Delta_{|\mc{Y}_r|} \rightarrow \left(\ve b_i, h_{|\mc{Y}_r|}(\ve b_i), g(\ve b_i) \right)$
Remember that $\Delta_{|\mc{Y}_r|}$ is $(|\mc{Y}_r|-1)$-dimensional, so the polytope $\Delta_{|\mc{Y}_r|} \times [0, \log(|\mc{Y}_r|)] \times [0, \log(|\mc{X}_1|\cdot |\mc{X}_2|)]$ is $(|\mc{Y}_r|+1)$-dimensional and the mapping assigns points inside this polytope for each $\ve b_i$.
Let $\mc S$ denote the set all all these points. As $h_{|\mc{Y}_r|}(\ve b_i)$ and $g(\ve b_i)$ are continuous functions of $\ve b_i$ \cite[Chapter 2.3]{yeung2008information}, $\mc S$ is compact and connected. 
\\Define $\mc C$ as the convex hull of $\mc S$, i.e. $\mc C = \text{conv}(\mc S)$.

\begin{lemma}
 The set of pairs $(\ve p_{y_r'}, \xi, \eta)$ defined in Eqs.~(\ref{eq:p_proof} - \ref{eq:eta_proof}) is precisely $\mc C$, for all integers ${\yhcard}>0$, $\ve{p}_{\hat{y}_r} \in \Delta_{{\yhcard}}$, $\ve b_i \in \Delta_{|\mc{Y}_r|}$, $i = 1,\ldots,{\yhcard}$.
\label{lemma:convexhull}
\end{lemma}
\begin{proof}
 \cite[Lemma 2.1]{witsenhausen1975conditional} by definition of the convex hull.
\end{proof}

\begin{prop}
 The function $F(x)$ is convex in $x$ for $0\leq x \leq H(Y_r)$.
\label{prop:F(x)convex}
\end{prop}

\begin{proof}
The function of interest $F(x)$ is the minimum of $\eta$ for which $\ve p_{y_r'} = \ve{p}_{y_r}$ and $\xi = x$. 

Define $\mc C_{\ve{p}_{y_r}}$ as the the projection of the intersection of $\mc C$ with the convex and compact set defined by $\ve p_{y_r'} = \ve{p}_{y_r}$ onto the plane $(\xi,\eta) \subset \mathbb{R}^2$.
Precisely, $\mc C_{\ve{p}_{y_r}} = \text{Proj}\{\mc C \cap \mc L_{\ve{p}_{y_r}} \}$, where $\mc L_{\ve{p}_{y_r}} = \{(\ve p_{y_r'},\xi,\eta) \subset \mathbb{R}^{|\mc{Y}_r|+1}| \ve p_{y_r'} = \ve{p}_{y_r}\}$. Note that $\mc L_{\ve{p}_{y_r}}$ is convex and compact.
As convexity is preserved under intersection \cite[Sect. 2.3.1]{boyd2004convex} and projection onto coordinates \cite[Sect. 2.3.2]{boyd2004convex}, the set $\mc C_{\ve{p}_{y_r}}$ is also convex and compact. That is, the infimum in Eq.~(\ref{eq:Def_F(x)}) can be attained and is thus a minimum, provided that the intersection $\mc C \cap \mc L_{\ve{p}_{y_r}} \not = \emptyset$. 
Like in \cite{witsenhausen1975conditional}, for ${\yhcard}=1$, $\ve b_1 = \ve{p}_{y_r}$, $p(\hat y_{r1})=1$ it follows that $\ve p_{y_r'} = \ve{p}_{y_r}$ and $\xi = h_{|\mc{Y}_r|}(\ve{p}_{y_r}) = H(Y_r)$. Choosing ${\yhcard}=|\mc{Y}_r|$, $\ve{p}_{\hat{y}_r} = \ve{p}_{y_r}$, $B = [\ve b_1, \ldots, \ve b_{|\mc{Y}_r|}] = I_{|\mc{Y}_r|}$, it follows that $\ve p_{y_r'} = \ve{p}_{y_r}$ and $\xi = 0$. By the convexity of $\mc C$, there must be points for which $\ve p_{y_r'} = \ve{p}_{y_r}$ and $\xi = x$ for $0\leq x \leq H(Y_r)$. Thus, the intersection is never empty.

As $\mc C_{\ve{p}_{y_r}}$ is convex and compact and $F(x)$ is the boundary of this convex set, $F(x)$ is itself convex, as illustrated in Fig.~\ref{fig:convex_set_proof} (similar to  \cite[Fig. C.2]{zeitler2012}).

\begin{figure}[ht]
 \centering
\psfrag{a}{$\xi$}
\psfrag{b}{$\eta$}
\psfrag{c}{${H(X_1|X_2,Y_r)+} \atop {H(X_2|X_1,Y_r)}$}
\psfrag{d}{$\log|\mc X_1| + \log|\mc X_2|$}
\psfrag{e}{$H(Y_r)$}
\psfrag{f}{increasing $C$}
\psfrag{g}{$\mc C_{\ve{p}_{y_r}}$}
\psfrag{h}{$F(x)$}
\psfrag{i}{$0$}
\includegraphics[width=0.4\columnwidth]{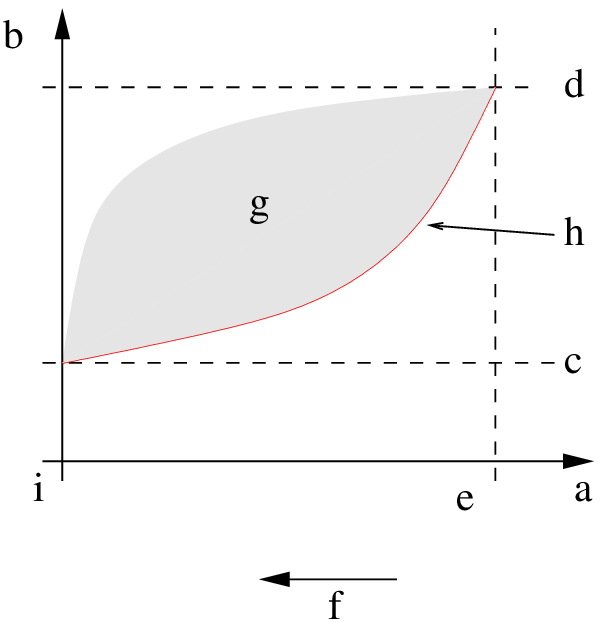}
\caption{Visualization of $\mc C_{\ve{p}_{y_r}}$ and the function $F(x)$}
\label{fig:convex_set_proof}
\end{figure}

\end{proof}

\begin{cor}
 $I(C)$ is a concave function in $C$, for $0\leq C \leq  H(Y_r)$.
\label{cor:I(C)_concave}
\end{cor}

 The proof follows by Proposition \ref{prop:F(x)convex}.

\subsection{Interpretation of Lagrangian}
It is known from Cor.~\ref{cor:I(C)_concave} that $I(C)$ is concave in $C$ corresponding to the quantizer rate $I(Y_r;\hY)$, for $0\leq C \leq H(Y_r)$.
Let $(C_\lambda, I_\lambda)$ be the point on the $I(C)$ curve at which the tangent has the slope $\lambda$. As $I(C)$ is nondecreasing, $\lambda \geq 0$.
The tangent at this point intersects with the y-axis at $I_\lambda - \lambda C_\lambda$, as illustrated in Fig.~\ref{fig:I(C)_concave}.
Recall that we abbreviate $p(\hy|y_r)$ by $p$:
Let $(C_{p'},I_{p'})$ be another point with $C_{p'} = I(Y_r;\hY)$ and $I_{p'} = J(p')$ corresponding to some conditional pmf (cond. pmf) $p' \not = p^*$, i.e. not lying on the $I(C)$-curve.
The line through $(C_{p'},I_{p'})$ with slope $\lambda$ intersects the y-axis at $I_{p'}-\lambda C_{p'}$.
Due to the concavity of $I(C)$, all intersections of lines with a given slope $\lambda$ lie below $I_\lambda - \lambda C_\lambda$.
To find the optimal axis intercept $I_\lambda - \lambda C_\lambda$ one can write:

\begin{align}
 I_\lambda - \lambda C_\lambda = \max_{\text{cond. pmf } p} \left\{ I_p - \lambda C_Q \right\} = \max_{\text{cond. pmf } p} \{ \mc L(p,\lambda)\}, \quad \lambda \geq 0
\end{align}
with $\mc L(p,\lambda)$ corresponding to the Lagrangian of the problem (dropping the equality constraints for $p$ which are captured in the constraint set here).

\begin{figure}[ht]
 \centering
\psfrag{a}{$C$}
\psfrag{b}{$I(C)$}
\psfrag{c}{$(I_{\lambda},C_{\lambda}) $}
\psfrag{d}{$I(X_1;Y_r|X_2) + \atop I(X_2;Y_r|X_1)$}
\psfrag{e}{$(I_{p'},C_{p'})$}
\psfrag{f}{\hspace*{-10mm}$I_{\lambda}-\lambda C_{\lambda}$}
\psfrag{g}{\hspace*{-12mm}$I_{p'}-\lambda C_{p'}$}
\psfrag{h}{$\lambda$}
\psfrag{i}{$0$}
\includegraphics[width=0.4\columnwidth]{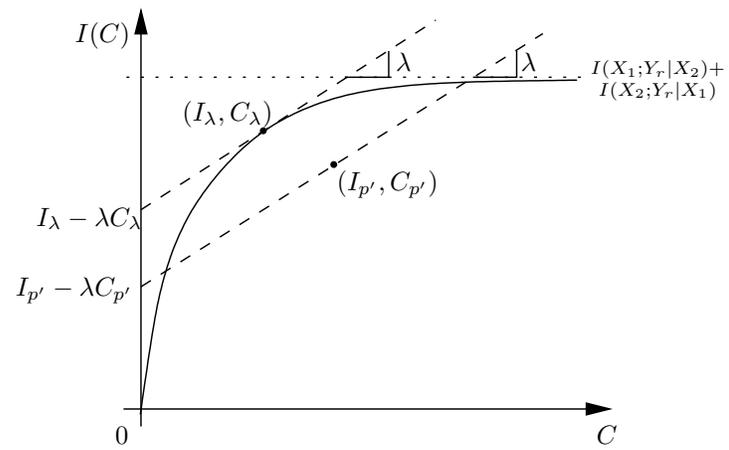}
\caption{Visualization of $I(C)$}
\label{fig:I(C)_concave}
\end{figure}

Running the optimization routine in Sec.~\ref{sec:computeI(C)} for a Lagrangian multiplier $\lambda$ should return the point on $I(C)$ with slope $\lambda$.

\end{document}